\documentclass[aps,prl,reprint,superscriptaddress,twocolumn,longbibliography]{revtex4-1}
\usepackage{amsfonts,amssymb,amscd,amsthm}
\usepackage{graphicx}
\usepackage{mathrsfs}
\usepackage[intlimits]{amsmath}
\usepackage[colorlinks, citecolor=red]{hyperref}
\usepackage{etoolbox}

\theoremstyle{definition}
\newtheorem*{theorem*}{Theorem}

\begin{document}
\title{Stretched Exponential Scaling of Parity-Restricted Energy Gaps in a Random Transverse-Field Ising Model}

\author{G.-X. Tang}
\thanks{These authors contribute equally to this work}%
\affiliation{Center for Quantum Information, Institute for Interdisciplinary Information Sciences, Tsinghua University, Beijing 100084, PR China}
\affiliation{Shanghai Qi Zhi Institute, AI Tower, Xuhui District, Shanghai 200232, PR China}

\author{J.-Z. Zhuang}
\thanks{These authors contribute equally to this work}%
\affiliation{Center for Quantum Information, Institute for Interdisciplinary Information Sciences, Tsinghua University, Beijing 100084, PR China}
\affiliation{Shanghai Qi Zhi Institute, AI Tower, Xuhui District, Shanghai 200232, PR China}

\author{L.-M. Duan}
\email{lmduan@tsinghua.edu.cn}
\affiliation{Center for Quantum Information, Institute for Interdisciplinary Information Sciences, Tsinghua University, Beijing 100084, PR China}
\affiliation{Hefei National Laboratory, Hefei 230088, PR China}

\author{Y.-K. Wu}
\email{wyukai@mail.tsinghua.edu.cn}
\affiliation{Center for Quantum Information, Institute for Interdisciplinary Information Sciences, Tsinghua University, Beijing 100084, PR China}
\affiliation{Hefei National Laboratory, Hefei 230088, PR China}

\begin{abstract}
The success of a quantum annealing algorithm requires a polynomial scaling of the energy gap. Recently it was shown that a two-dimensional transverse-field Ising model on a square lattice with nearest-neighbor $\pm J$ random coupling has a polynomial energy gap in the symmetric subspace of the parity operator [Nature 631, 749-754 (2024)], indicating the efficient preparation of its ground states by quantum annealing. However, it is not clear if this result can be generalized to other spin glass models with continuous or biased randomness. Here we prove that under general independent and identical distributions (i.i.d.) of the exchange energies, the energy gap of a one-dimensional random transverse-field Ising model follows a stretched exponential scaling even in the parity-restricted subspace. We discuss the implication of this result to quantum annealing problems.
\end{abstract}

\maketitle

Ising model is one of the best-known physical models in statistical mechanics. Originating from the study of ferromagnetism \cite{LenzBeitragZV,Ising1925,PhysRev.65.117}, it has now become a testing ground for understanding classical and quantum phase transitions \cite{sachdev2000quantum}, and is widely used in fields like neural networks \cite{ACKLEY1985147} and optimization \cite{kirkpatrick1983optimization,10.3389/fphy.2014.00005}. In particular, with disorder and geometric frustrations in the interaction of the Ising model, a spin glass phase can appear which possesses a complicated energy landscape and a quasi-random ground state without translational symmetry \cite{RevModPhys.58.801,Sherrington2025}. It turns out that finding the ground state of such a general Ising model is NP-hard, as many well-known NP-complete problems can be mapped to the ground state configuration of an Ising model with suitable exchange energies \cite{Barahona_1982,10.1145/335305.335316,10.3389/fphy.2014.00005}.

Although it is widely believed that an NP-complete problem will not be efficiently solved even by a quantum computer, several quantum algorithms such as quantum annealing \cite{PhysRevE.58.5355,doi:10.1126/science.284.5415.779,RevModPhys.80.1061,Yarkoni_2022} and the quantum approximate optimization algorithm (QAOA) \cite{farhi2014quantumapproximateoptimizationalgorithm,PhysRevX.10.021067} have been developed and applied to solve the ground state of the Ising spin glass, in the hope that speedup can still be achieved for special classes of problems relevant to physics and industry. While a systematic understanding of the performance of QAOA is challenging \cite{PhysRevX.10.021067}, the success of quantum annealing or adiabatic quantum computation requires an annealing timescale of at least $O(1/\Delta^2)$ where $\Delta$ is the energy gap between the ground state and the first excited state \cite{RevModPhys.80.1061,RevModPhys.90.015002}, which may be improved to $O(1/\Delta)$ considering nonadiabatic paths \cite{PhysRevA.84.012312,boixo2010fastquantumalgorithmstraversing}. Therefore, a critical question for efficiently preparing the ground state of certain Ising models by quantum annealing, or efficiently solving certain families of combinatorial optimization problems, is whether the energy gap decays polynomially or exponentially with the increasing system size.

A random transverse-field Ising model is a prototypical model for understanding the critical behavior of a spin glass, with its exchange energies and/or the transverse fields drawn independently from certain random distributions \cite{RevModPhys.58.801,PhysRevLett.69.534,PhysRevB.51.6411,PhysRevB.53.8486}. Usually two types of random distributions are considered: a continuous Gaussian distribution as in the Edwards-Anderson model \cite{Edwards_1975,RevModPhys.58.801}, and a discrete distribution as in the $\pm J$ model or the frustration model \cite{Vannimenus_1977,RevModPhys.58.801}. In both cases, it has been shown that the energy gap $\Delta$ follows an activated (stretched exponential) scaling $\ln(1/\Delta)\sim L^\psi$ where $L$ is the length scale of the system and $\psi$ is a critical exponent with $\psi=1/2$ in one dimension (1D) \cite{PhysRevLett.69.534,PhysRevB.51.6411,PhysRevB.53.8486,PhysRevB.58.9131,FISHER1999222}, $\psi\approx 0.48$ in two dimensions (2D) \cite{PhysRevB.82.054437,PhysRevLett.81.5916,PhysRevB.61.1160}, and $\psi\approx 0.46$ in three and four dimensions \cite{PhysRevB.83.174207}. Note that the exponential scaling of the energy gap does not exclude the possibility that more efficient quantum algorithms can be developed, e.g. by a sophisticated design of the quantum annealing path from different initial states. Indeed, there exist polynomial algorithms to solve the ground state of a 1D classical Ising model with nearest-neighbor interactions, or a 2D classical Ising model on a planar graph \cite{Barahona_1982,plummer1986matching}.

Recently, Ref.~\cite{Bernaschi2024} shows that, despite the overall exponential energy gap between the ground state and the first excited state, the energy gap relevant for quantum annealing may remain algebraic if the evolution path is restricted to the subspace of parity symmetry. It is thus suggested that a standard quantum annealing process may still be able to prepare the ground state of a wide family of spin glass problems efficiently as long as the parity symmetry is protected. In that work, a 2D square lattice with $\pm J$ random coupling between nearest neighbors is considered. This leads to the question whether the same conclusion holds for more general problems with continuous or biased randomness, as in many practical optimization problems. In this work, we answer this question by considering a 1D random transverse-field Ising model. We prove analytically that its energy gap follows an activated scaling even in the parity symmetry subspace whenever the independently sampled exchange energy $|J|$ has a nonzero standard deviation. During the proof, we also obtain an analytical formula to upper-bound this parity-restricted energy gap, which is more robust against numerical errors than the direct diagonalization of a free-fermion Hamiltonian \cite{LIEB1961407,PhysRev.127.1508,PFEUTY197079,PhysRevB.53.8486}. With this tool, we further show that as the system size increases, the level of randomness that can be tolerated for the exponential scaling to be subdominant will decrease accordingly, raising the demand for a fault-tolerant device for large-scale problems even with a polynomial energy gap in the ideal case. We also note that our assumption of the Ising coupling being independently sampled, although widely used in the literature, may not be necessary for the occurrence of the activated scaling of the energy gap. We explicitly construct a deterministic sequence with physical motivations and numerically show that its parity-restricted energy gap is bounded by an activated scaling up to large system sizes.

\emph{Model and main results.} Throughout this work, we consider a 1D transverse-field Ising model with open boundary conditions
\begin{equation}
H = -\sum_{i=1}^{L-1} J_{i,i+1} X_i X_{i+1} - h \sum_{i=1}^L Z_i,
\end{equation}
where $X_i$ and $Z_i$ represent the Pauli operators for the $i$th spin. This Hamiltonian commutes with the parity operator $P\equiv \prod_{i=1}^L Z_i$, so that we can assign an even ($P=+1$) or odd ($P=-1$) parity to each energy eigenstate. We allow the exchange energy $J_{i,i+1}$ to take arbitrary values as in general optimization problems, and we apply a uniform transverse field $h$ as in the quantum annealing setting.

Suppose the exchange energies $J_{i,i+1}$ are drawn from independent and identical distributions (i.i.d.) with a positive standard deviation $\mathrm{std}(\ln|J|)>0$. Note that for a 1D Ising model with nearest-neighbor interactions, we can always flip the spins sequentially (by applying a unitary transform $Z_i$) from one side to the other to set their interactions to be positive without changing the level structures, so the sign of $J_{i,i+1}$ does not matter to the energy gap. Also here we assume $J_{i,i+1}\ne 0$ ($\forall i$) or that it has a probability measure of zero, because otherwise we are basically just dividing a long chain into smaller segments.

The critical point of the above 1D random transverse-field Ising model (RTIM) is at $\ln |h_c|=\overline{\ln|J|}$ \cite{PhysRevB.51.6411,PhysRevB.53.8486,PhysRevB.58.9131} where the overline represents an average over random realizations. Our main result can be summarized as the following theorem.

\begin{figure}[!tbp]
   \includegraphics[width=\linewidth]{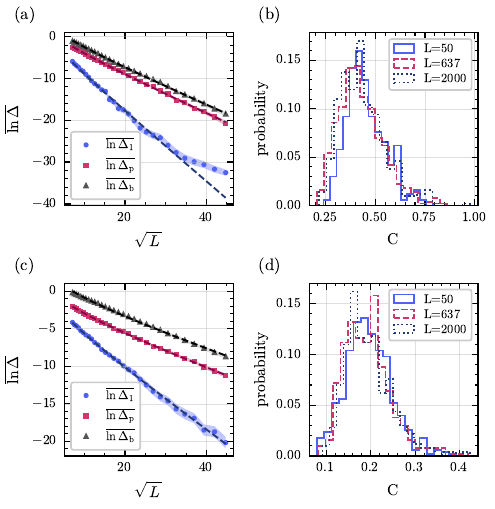}
   \caption{Numerical results for the activated scaling of RTIM's energy gaps. (a) The standard energy gap $\Delta_1$ (blue dots), the parity-restricted energy gap $\Delta_p$ (red squares) and the derived upper bound $\Delta_b$ (black triangles) versus the system size $L$, when the random Ising couplings are independently drawn from a Gaussian distribution $\mathcal{N}(\mu,\sigma^2)$ with a mean value $\mu=1$ and a standard deviation $\sigma=0.5$. The transverse field $h$ is fixed at the critical point $h_c=\exp(\overline{\ln|J|})$ where the average over the random distribution is estimated from a sample size of ten times the largest system size. Here we plot the energy gaps in logarithmic scale and the system size in square root scale to visualize the activated scaling as a straight line. Each data point is the geometric average (typical gap \cite{PhysRevB.58.9131}) over 500 random realizations, with the shaded area representing 99\% confidence intervals. The dashed lines are linear fitting results, where we have dropped the data points for $\Delta_1<10^{-13}$ to suppress the numerical error. (b) The histogram of $C\equiv L^{-1/2}\ln\Delta_p$. The three curves for $L=50,\,637,\,2000$ collapse well with each other. (c) and (d) Similar plots when the random Ising couplings are independently drawn from a two-point distribution $\{0.5,\,1\}$ with equal probabilities.
   \label{fig:exponential_gap}}
\end{figure}

\begin{theorem*}
At the critical point $h=h_c$, with high probability, the parity-restricted energy gap of the above 1D RTIM will be bounded by an activated scaling. Specifically, for any targeted failure probability $\epsilon>0$ there exists a constant $c>0$ independent of the system size $L$ such that $\lim_{L\to \infty} \mathrm{Pr}[\Delta_p(L)\le e^{-c\sqrt{L}}]\ge 1-\epsilon$, where $\Delta_p(L)$ is the energy gap of a 1D RTIM of a size $L$ between the ground state and the first excited state with the same parity.
\end{theorem*}

\begin{proof}[Proof sketch]
Here we give a sketch of the proof, while the details can be found in Supplemental Material. (1) It is well-known that a 1D RTIM can be mapped to a free-fermion model by Jordan-Wigner transformation, and its eigenstates can be described by the excitation of different fermionic modes with energies $\epsilon_1\le \epsilon_2\le\cdots$ \cite{LIEB1961407,PhysRev.127.1508,PFEUTY197079,PhysRevB.53.8486}. Note that each fermionic excitation flips the parity of the eigenstate, therefore the standard energy gap is given by $\Delta_1=\epsilon_1$ while the parity-restricted energy gap by $\Delta_p=\epsilon_1+\epsilon_2$. It has already been proved that $\epsilon_1$ follows an activated scaling $e^{-c_1\sqrt{L}}$ \cite{PhysRevLett.69.534,PhysRevB.51.6411,PhysRevB.53.8486,PhysRevB.58.9131,FISHER1999222}, so for our purpose it suffices to show that $\epsilon_2$ follows a similar scaling. (2) The eigenvalues of the free-fermion Hamiltonian can further be expressed as a random walk with a reflecting (Neumann) boundary at one side and an absorbing (Dirichlet) boundary at the other side \cite{PhysRevB.106.064204}. The transition matrix $T$ of this random walk has eigenvalues $\lambda_i=\epsilon_i^2$ ($i=1,\,2,\,\cdots$). (3) As a standard procedure for a reversible Markov chain, the eigenvalues $\lambda_i$'s of the transition matrix $T$ can be written as a Rayleigh quotient \cite{aldous-fill-2014}, so the second smallest eigenvalue $\lambda_2=\epsilon_2^2$ can be formulated by the Courant-Fischer theorem as a minimization of the maximal Rayleigh quotient over all the two-dimensional subspaces. In particular, we can break the length-$L$ chain into two segments with sizes $L_1$ and $L_2$ ($L_1+L_2=L$), and consider a two-dimensional subspace spanned by two vectors which are supported on the two segments, respectively. It thus provides an upper bound $\lambda_2(L)\le 2\max\{\lambda_1^{ND}(L_1),\lambda_1^{DD}(L_2)\}$ from the smallest eigenvalues for these two segments, where the superscripts $N$ and $D$ represent the Neumann and Dirichlet boundary conditions, respectively. (Note that when breaking a chain into two parts, we create Dirichlet boundaries in the middle.) (4) The first term in the maximum follows the previously proved activated scaling, while the second term can again be bounded by the Courant-Fischer theorem using a suitably constructed trial vector. Combining all these results, we obtain an activated scaling for the parity-restricted energy gap $\Delta_p$ versus the system size $L$.
\end{proof}

\begin{figure}[!tbp]
   \includegraphics[width=\linewidth]{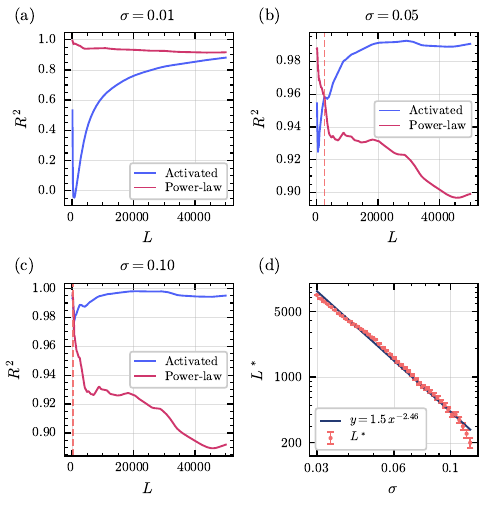}
   \caption{Quality of different fitting models versus level of random fluctuation. After obtaining the upper bound of the parity-restricted energy gap $\Delta_b$ versus the system size $L$ as in Fig.~\ref{fig:exponential_gap}(a), we can linearly fit them by either the $\ln \Delta_b$ vs. $\sqrt{L}$ model or the $\ln \Delta_b$ vs. $\ln L$ model, so as to numerically distinguish between an activated scaling and a polynomial scaling. Here we choose three typical standard deviations of the Gaussian distribution in (a) $\sigma=0.01$, (b) $\sigma=0.05$ and (c) $\sigma=0.1$, and examine the quality of these two fitting models versus different system sizes $L$ up to which the fitting is performed. For small random fluctuation, initially the polynomial fitting can give higher $R^2$ values and the activated scaling outperforms only when the system size gets above certain value $L^*$ [indicated by the vertical dashed lines in (b) and (c)]. (d) The critical system size $L^*$ after which the activated scaling dominates versus the random fluctuation level $\sigma$. The energy gaps are evaluated at system sizes with an increment of $\Delta L=50$, so we estimated an error bar of $\pm \Delta L/2$ for the extracted $L^*$. A power-law relation $L^*\sim \sigma^{-2.46}$ is fitted.
   \label{fig:scaling}}
\end{figure}

\emph{Numerical results and discussions.} To further illustrate this activated scaling and the validity of our upper bound, we present numerical results of the energy gaps for different types of i.i.d. distributions in Fig.~\ref{fig:exponential_gap}. Throughout this work, we consider typical energy gaps as the geometric mean over random realizations \cite{PhysRevB.58.9131}. In the upper (lower) panel, we choose a continuous (discrete) distribution for the Ising couplings. The exact values of the standard energy gap $\Delta_1$ and the parity-restricted energy gap $\Delta_p$ are obtained by mapping to a free-fermion model \cite{PhysRevB.106.064204}, while the upper bound $\Delta_b$ for the parity-restricted gap is calculated using the formula in Supplemental Material. For both continuous and discrete distributions, we obtain good linear relations on the $\ln \Delta$ vs. $\sqrt{L}$ plots [apart from the small-$\Delta_1$ regime in Fig.~\ref{fig:exponential_gap}(a) which is limited by the numerical error in the singular value decomposition], indicating an activated scaling for both the standard energy gap $\Delta_1$ (blue dots) and the parity-restricted energy gap $\Delta_p$ (red squares), although the slope for $\Delta_1$ is generally steeper than that for $\Delta_p$. This scaling is further confirmed by the collapse of the histogram of $C\equiv L^{-1/2}\ln\Delta_p$ for different system sizes as shown in Fig.~\ref{fig:exponential_gap}(b) and (d) over 500 random realizations. Also note that our analytical upper bound (black triangles) can correctly capture the tendency of $\Delta_p$ in Fig.~\ref{fig:exponential_gap}(a) and (c), while its numerical stability allows us to push the calculation to larger system sizes with an energy gap below the machine precision.

Note that our theorem about a stretched exponential scaling of the parity-restricted energy gap holds for arbitrarily small randomness in $|J|$. This is consistent with the Harris criterion \cite{Harris_1974,Luck1993} which states that an arbitrarily weak disorder can be relevant to the critical behavior if $(d+z)\nu < 2$ where $d$ is the spatial dimension of the system, $z$ is the dynamical exponent, and $\nu$ is the critical exponent of the correlation length. Actually, previously it has been shown that for $d\le 4$ the critical point of the transverse-field Ising model is an infinite-disorder fixed point \cite{PhysRevB.83.174207,PhysRevB.61.1160}, where the random fluctuation gets enlarged as the system is coarse grained. Here we study the implication of this result for the parity-restricted energy gap on solving quantum annealing problems with an analogue quantum computer. We consider an RTIM with its exchange energies drawn randomly and independently from a Gaussian distribution $\mathcal{N}(\mu=1,\sigma^2)$. Ideally with $\sigma=0$, the homogeneous transverse-field Ising model has its parity-restricted energy gap decaying polynomially with the system size at the critical point $h_c=1$. Now for nonzero $\sigma$'s, we compute the upper bound $\Delta_b$ of the parity-restricted energy gap up to a large system size $L=50000$, and plot the quality of the polynomial and the stretched exponential fitting models in Fig.~\ref{fig:scaling} for (a) $\sigma=0.01$, (b) $\sigma=0.05$ and (c) $\sigma=0.1$. As expected, for small $\sigma$ and small $L$, the energy gap is better fitted by a polynomial scaling, but as the system size $L$ increases, the stretched exponential scaling will finally prevail. In Fig.~\ref{fig:scaling}(d) we plot the critical system size $L^*$ after which the activated scaling starts to dominate, and observe a power-law relation $L^*\sim \sigma^{-2.46}$. This also suggests that, as the problem size $L$ increases, the random noise or inaccuracy that can be tolerated by an analogue quantum simulator in engineering the desired Hamiltonian decreases as $L^{-0.4}$. A noise level above this value will lead to an exponential energy gap, and therefore generally an exponentially small probability to find the ground state during a polynomial annealing time, even if the ideal problem to be solved has a provable polynomial energy gap. Therefore, at large scale it may still be necessary to run the quantum annealing algorithm on an error-corrected quantum computer, although generally it is expected that quantum annealing has some intrinsic robustness to certain types of errors.

\begin{figure}[!tbp]
   \includegraphics[width=\linewidth]{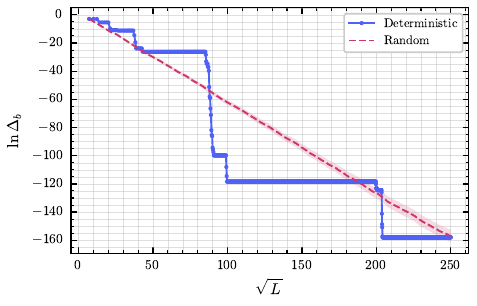}
   \caption{A deterministic sequence with a parity-restricted energy gap showing an activated scaling. Here we consider a deterministic sequence $J_{x,x+1}=\sin(x \sin x)$ ($x=1,\,2,\,\cdots,\,L-1$), and plot the upper bound $\Delta_b$ of the parity-restricted energy gap. Similar to a single random realization of the RTIM, multiple plateaus are observed along with the overall activated scaling as the system size $L$ increases. In comparison, we also plot $\Delta_b$ for an RTIM whose exchange energies are sampled from a uniform distribution between $[-1,\,1]$ (red dashed). The data points for the random distribution are the geometric mean over $500$ realizations, with the shaded area representing 99\% confidence intervals.
   \label{fig:deterministic}}
\end{figure}

Finally, we note that although the above results are obtained for random Ising couplings following i.i.d. distributions, it is also possible to observe similar activated scaling for Ising couplings generated by a deterministic sequence. Previously it has been shown that a stretched exponential scaling for the standard energy gap $\Delta_1$ can occur when the sequence shows ``unbounded fluctuation'', and several sequences have been constructed by repetitively applying some substitution rules which corresponds to a self-similar structure \cite{Luck1993,Igloi1998}. Here we construct a different type of deterministic sequence $J_{x,x+1}=\sin(x \sin x)$ ($x=1,\,2,\,\cdots$). The physical motivation for this sequence comes from performing quantum gates in series in a digital trapped-ion quantum simulator. For example, we may have a linear chain of ions with a spacing of $d$, and we may choose to sequentially perform two-qubit entangling gates between nearest neighbors from one side to the other to simulate the Ising interaction through Trotter decomposition. Now imagine that there is a small change in the axial trapping potential, which leads to a $\delta$ shift in the ion spacing. Then the position drift of the $n$-th ion will be $n\delta$ with respect to the focused addressing laser beam. Furthermore, suppose the change in the axial potential itself is oscillating, which may come from some low-frequency electrical noise. Then we may have $\delta\sim \sin(\omega t) \sim \sin(c n)$ as we address the ions sequentially. Finally, the position drift with respect to the addressing laser beam can cause a change in the Rabi rate and hence an error in the two-qubit entanglement phase, causing the $\sin(x \sin x)$-type sequence in the effective Ising coupling. In Fig.~\ref{fig:deterministic} we plot the upper bound $\Delta_b$ of the parity-restricted energy gap for this deterministic sequence. Different from the smooth curves for random distributions after averaging over a large number of realizations, here we observe several plateaus as the system size $L$ increases, which actually also occurs when examining individual realizations of random distributions. Nevertheless, the parity-restricted energy gap for this deterministic sequence still follows a similar tendency as that of an RTIM whose Ising couplings are drawn uniformly between $[-1,\,1]$, and we observe that the parity-restricted gaps for both models decay to the exponentially small value of $10^{-69}$ at the system size of $L=62500$.

\emph{Conclusion}. To sum up, we have proven the super-polynomial energy gap in the parity symmetric subspace for a 1D RTIM with general randomness. This suggests that the standard annealing path starting from the easily prepared ground state of the transverse field cannot efficiently lead to the ground state of the desired final spin glass Hamiltonian. Note that in the special cases of 1D nearest-neighbor or 2D planar graphs, there exist polynomial classical algorithms to solve the ground state of the spin glass \cite{Barahona_1982,plummer1986matching}, so our result does not exclude the possibility that efficient quantum annealing paths can still be designed starting from some easily prepared ground states of different Hamiltonians. However, such designs will likely need to take into account the special structure of the 1D or 2D systems, because it is well-known that solving the ground state of a general spin glass is NP-hard \cite{Barahona_1982,10.1145/335305.335316,10.3389/fphy.2014.00005} and is not likely to be efficient even on a universal quantum computer.

\begin{acknowledgments}
This work was supported by the Quantum Science and Technology-National Science and Technology Major Project (2021ZD0301601), the Shanghai Qi Zhi Institute, the Tsinghua University Initiative Scientific Research Program, and the Ministry of Education of China. L.-M. D. acknowledges in addition support from the New Cornerstone Science Foundation through the New Cornerstone Investigator Program. Y.-K. W. acknowledges in addition support from Tsinghua University Dushi program.
\end{acknowledgments}

%

\end{document}


\title{Supplemental Material for
``Stretched Exponential Scaling of Parity-Restricted Energy Gaps in a Random Transverse-Field Ising Model''}

\author{G.-X. Tang}
\thanks{These authors contribute equally to this work}%
\affiliation{Center for Quantum Information, Institute for Interdisciplinary Information Sciences, Tsinghua University, Beijing 100084, PR China}
\affiliation{Shanghai Qi Zhi Institute, AI Tower, Xuhui District, Shanghai 200232, PR China}

\author{J.-Z. Zhuang}
\thanks{These authors contribute equally to this work}%
\affiliation{Center for Quantum Information, Institute for Interdisciplinary Information Sciences, Tsinghua University, Beijing 100084, PR China}
\affiliation{Shanghai Qi Zhi Institute, AI Tower, Xuhui District, Shanghai 200232, PR China}

\author{L.-M. Duan}
\email{lmduan@tsinghua.edu.cn}
\affiliation{Center for Quantum Information, Institute for Interdisciplinary Information Sciences, Tsinghua University, Beijing 100084, PR China}
\affiliation{Hefei National Laboratory, Hefei 230088, PR China}

\author{Y.-K. Wu}
\email{wyukai@mail.tsinghua.edu.cn}
\affiliation{Center for Quantum Information, Institute for Interdisciplinary Information Sciences, Tsinghua University, Beijing 100084, PR China}
\affiliation{Hefei National Laboratory, Hefei 230088, PR China}

\maketitle

\makeatletter
\renewcommand{\thefigure}{S\arabic{figure}}
\renewcommand{\thetable}{S\arabic{table}}
\renewcommand{\theequation}{S\arabic{equation}}
\makeatother

\section{Proof for activated scaling of parity-restricted energy gap}
\label{sec:proof}
In this section we provide the complete proof for the theorem in the main text, which is restated here for convenience.

\noindent\textbf{Model}. We consider a 1D transverse-field Ising model with open boundary conditions
\begin{equation}
H = -\sum_{i=1}^{L-1} J_i X_i X_{i+1} - h \sum_{i=1}^L Z_i, \label{eq:RTIM}
\end{equation}
whose exchange energies $J_i$'s are nonzero and are drawn from independent and identical distributions with a positive standard deviation $\mathrm{std}(\ln|J|)>0$. Its critical point is at $\ln |h_c|=\overline{\ln|J|}$. The parity operator is defined as $P\equiv \prod_{i=1}^L Z_i$ and commutes with the Hamiltonian.

\begin{theorem}[Stretched-exponential bound for the parity-restricted energy gap]\label{thm:main}
At the critical point $h=h_c$, with high probability, the parity-restricted energy gap of the above 1D RTIM will be bounded by an activated scaling. Specifically, for any targeted failure probability $\epsilon>0$ there exists a constant $c>0$ independent of the system size $L$ such that $\lim_{L\to \infty} \mathrm{Pr}[\Delta_p(L)\le e^{-c\sqrt{L}}]\ge 1-\epsilon$, where $\Delta_p(L)$ is the energy gap of a 1D RTIM of a size $L$ between the ground state and the first excited state with the same parity.
\end{theorem}

Below we present the detailed proof.

\subsection{Mapping to a Markov process}
Through the Jordan-Wigner transformation and a Bogoliubov rotation, the Hamiltonian in Eq.~(\ref{eq:RTIM}) can be diagonalized into free fermions $H=\sum_{k=1}^{L}\epsilon_k (c_k^\dagger c_k-1/2)$, where the energies $0\le\epsilon_1\le\epsilon_2\le\cdots\le\epsilon_L$ of individual fermionic modes are given by the singular values of \cite{PhysRevB.106.064204}
\begin{equation}
M = 2\begin{pmatrix}
 h \\
 J_1 & h \\
  & J_2 & \ddots \\
  &  & \ddots & h
\end{pmatrix}. \label{eq:svd}
\end{equation}

Under this fermionic representation, the parity operator can be expressed as $P=(-1)^{N_f}$ where $N_f=\sum_{k=1}^{L}c_k^\dagger c_k$. The energy spectrum thus splits into two sectors of even and odd parities.
The standard energy gap is $\Delta_1=\epsilon_1$ and the parity-restricted energy gap is $\Delta_p=\epsilon_1+\epsilon_2$, since a single-particle excitation changes the parity.

Furthermore, the singular value decomposition problem in Eq.~(\ref{eq:svd}) can be related to a Markov process with the rate matrix \cite{PhysRevB.106.064204}
\begin{equation}
T=\begin{pmatrix}
 -h^{2} & h^{2} & 0 & \cdots & 0 \\
 J_1^{2} & -h^{2}-J_1^{2} & h^{2} & \ddots & \vdots \\
 \vdots & \ddots & \ddots & -h^{2}-J_{L-2}^{2} & h^{2} \\
 0 & \cdots & 0 & J_{L-1}^{2} & -h^{2}-J_{L-1}^{2}
\end{pmatrix}.
\label{eq:T_generator_ND}
\end{equation}
In particular, the eigenvalues of $-T$ are given by $\lambda_k=\epsilon_k^{2}/4$ ($k=1,\,\cdots,\,L$).
Note that all but the last row of $T$ sum up to zero, while its last row adds up to a negative value, hence $T$ forms a sub-generator of a continuous-time Markov process. In particular, note that $T$ is a tridiagonal matrix, so this Markov process belongs to the birth-death process where an integer population can change by at most one in an infinitesimal time step.

Now we borrow some standard notations from the birth-death process for the following analysis \cite{karlin1957classification}.
Let $p_i\equiv T_{i,i+1}=h^2\mathbbm{1}_{\{i<L\}}$ be the right-jump rate, $q_i\equiv T_{i,i-1}=J_{i-1}^2\mathbbm{1}_{\{i>1\}}$ be the left-jump rate, and $\kappa_i\equiv -\sum_{j}T_{ij}=h^2\mathbbm{1}_{\{i=L\}}$ be the absorption rate at the site $i$.
The current model has absorption only at the right boundary, which is referred to as the Neumann-Dirichlet (ND) boundary conditions. Later we will also consider models with absorption imposed at both ends, which is denoted as the Dirichlet-Dirichlet (DD) boundary conditions.

Ignoring the absorption at the boundary, we can define a new Markov process which is ergodic. Therefore, we can construct a unique equilibrium distribution $\pi=\{\pi_i\}_{i=1}^L$ for this new Markov process, satisfying the detailed balance condition $\pi_i p_i=\pi_{i+1} q_{i+1}$ ($i=1,\,\cdots,\,L-1$). Then we can define the edge conductances as $C_i \equiv \pi_i p_i \equiv \pi_{i+1}q_{i+1}$ in the bulk ($1\le i\le L-1$). The conductances at the ends are defined as $C_0\equiv \pi_1\kappa_1$ and $C_L\equiv \pi_L \kappa_L$, which are nonzero only for Dirichlet-type boundaries. The equilibrium distribution $\pi$ also allows us to define an inner product between any functions $f$ and $g$ on the $L$ sites ($L$-dimensional real vectors) as $\langle f,\,g\rangle_{\pi} \equiv \sum_{i=1}^L \pi_i f_i g_i$. This further leads to a Dirichlet form associated with the matrix $-T$
\begin{equation}
\label{eq:def_dirichlet}
\mathcal{E}(f,\,g) \equiv \langle f,\,-T g\rangle_{\pi} = \sum_{i=1}^{L-1} C_i (f_{i+1}-f_i)(g_{i+1}-g_i)
+ \sum_{i=1}^{L} \kappa_i \pi_i f_i g_i,
\end{equation}
which is apparently symmetric and positive semi-definite for the input vectors. Note that this result holds for both the ND and DD boundary conditions which will be used later.

Finally, we define the Rayleigh quotient $\mathcal{R}[f]\equiv \mathcal{E}(f,\,f)/\langle f,\,f\rangle_{\pi}$ for any nonzero $L$-dimensional vector $f$. The eigenvalues of $-T$ can thus be given by the Courant-Fischer min-max theorem \cite{Horn_Johnson_1985}
\begin{equation}
\lambda_k = \min_{\substack{V\subset \mathbb{R}^L\\ \dim V=k}} \max_{0\neq f\in V} \mathcal{R}[f] \qquad (k=1,\,\cdots,\,L).
\label{eq:CF_thm}
\end{equation}

\subsection{Upper bound of second smallest eigenvalue by breaking into subchains}
As described in the main text, it is well-known that the standard energy gap $\Delta_1=\epsilon_1$ of the 1D RTIM follows an activated scaling. To understand the scaling the parity-restricted energy gap $\Delta_p=\epsilon_1+\epsilon_2$, we thus need to analyze $\epsilon_2$, or equivalently $\lambda_2$ through the Courant-Fischer theorem. Specifically, we will show that $\lambda_2$ can be bounded by cutting the whole chain into two pieces and computing the $\lambda_1$'s of sub-chains.

\begin{proposition}[Two-block upper bound for the second smallest eigenvalue]
Let $T$ be the tridiagonal rate matrix defined above with its second smallest eigenvalue $\lambda_2$. For $m\in\{1,\,\cdots,\,L-1\}$, let $\lambda^{ND}_1([1,\,m])$ be the smallest eigenvalue on the left block $[1,\,m]$ (reflecting at the site $1$ and absorbing at the site $m$), and let $\lambda^{DD}_1([m+1,\,L])$ be the smallest eigenvalue on the right block $[m+1,\,L]$ (absorbing at both sites $m+1$ and $L$). Then
\begin{equation}
\lambda_2 \le 2\times \min_{1\le m\le L-1}\ \max \left\{\lambda^{ND}_1([1,m]),\, \lambda^{DD}_1([m+1,L])\right\}. \label{eq:two_blocks}
\end{equation}
\label{thm:block_upper_bound}
\end{proposition}

\begin{proof}
Fix $m\in\{1,\,\cdots,\,L-1\}$, and let $\phi^L$ and $\phi^R$ denote the eigenfunctions corresponding to the smallest eigenvalues on the left block $[1,\,m]$ and the right block $[m+1,\,L]$, respectively. We further denote their zero-extensions to the whole chain as $\psi^L$ and $\psi^R$, with $\psi^L(i)=\phi^L(i)$ for $i\le m$ and $0$ otherwise, and $\psi^R(i)=\phi^R(i)$ for $i\ge m+1$ and $0$ otherwise. Define the vector space $U=\mathrm{span}\{\psi^L,\,\psi^R\}$ spanned by the two vectors. The Courant-Fischer theorem gives
\begin{equation}
\lambda_2
=\min_{\substack{V\subset \mathbb{R}^L\\ \dim V=2}} \max_{0\neq f\in V} \mathcal{R}[f]
\le \max_{0\ne f\in U}\mathcal{R}[f].
\end{equation}

For general $f=a\psi^L+b\psi^R\in U$, we have $\langle f,\,f\rangle_\pi
= a^2\langle\psi^L,\,\psi^L\rangle_\pi + b^2\langle\psi^R,\,\psi^R\rangle_\pi$.
Also we have
\begin{equation}
\begin{aligned}
\mathcal{E}(f,\,f)
&= a^2\mathcal{E}(\psi^L,\,\psi^L)+b^2\mathcal{E}(\psi^R,\,\psi^R)-2ab\,C_m\,\psi^L(m)\psi^R(m{+}1)\\
&\le a^2\mathcal{E}(\psi^L,\,\psi^L)+b^2\mathcal{E}(\psi^R,\,\psi^R)+a^2C_m\psi^L(m)^2+b^2C_m\psi^R(m{+}1)^2\\
&\le 2a^2\mathcal{E}(\psi^L,\,\psi^L)+2b^2\mathcal{E}(\psi^R,\,\psi^R),
\end{aligned}
\end{equation}
where in the last line we have used the fact that $\mathcal{E}(\psi^L,\,\psi^L)$ is a sum over non-negative terms including $C_m \psi^L(m)^2$, and similarly for $\mathcal{E}(\psi^R,\,\psi^R)$.

Then we have
\begin{equation}
\begin{aligned}
\mathcal{R}[f]
=\frac{\mathcal{E}(f,\,f)}{\langle f,f\rangle_\pi}
&\le\
2\times\frac{a^2\mathcal{E}(\psi^L,\,\psi^L)+b^2\mathcal{E}(\psi^R,\,\psi^R)}{a^2\langle\psi^L,\,\psi^L\rangle_\pi + b^2\langle\psi^R,\,\psi^R\rangle_\pi}\\
&\le 2\max \left\{\frac{\mathcal{E}(\psi^L,\,\psi^L)}{\langle\psi^L,\psi^L\rangle_\pi},\, \frac{\mathcal{E}(\psi^R,\,\psi^R)}{\langle\psi^R,\psi^R\rangle_\pi}\right\}\\
&= 2\max \left\{\frac{\mathcal{E}_L(\psi^L,\,\psi^L)}{\langle\psi^L,\psi^L\rangle_{\pi_L}},\, \frac{\mathcal{E}_R(\psi^R,\,\psi^R)}{\langle\psi^R,\psi^R\rangle_{\pi_R}}\right\}\\
&=2\max\left\{\lambda^{ND}_1([1,\,m]),\, \lambda^{DD}_1([m+1,\,L])\right\},
\end{aligned}
\end{equation}
where in the second line we use the inequality $(a+b)/(c+d)\le \max\{a/c,\,b/d\}$ when $a,\,b,\,c,\,d>0$.
In the third line, we note that when restricting the equilibrium distribution $\pi$ to the two blocks, it just gives the corresponding sub-chain equilibrium distributions $\pi_L$ and $\pi_R$ up to a normalization. Similarly, the Dirichlet form $\mathcal{E}$ in Eq.~(\ref{eq:def_dirichlet}) simply maps to the sub-chain ones $\mathcal{E}_L$ and $\mathcal{E}_R$ with the bulk terms $C_m \psi^L(m)^2$ and $C_m \psi^R(m+1)^2$ mapped to the boundary terms $\kappa_m \pi_m \psi^L(m)^2$ and $\kappa_{m+1} \pi_{m+1} \psi^R(m+1)^2$, respectively.

Taking the maximum over $f\in U$ and then the minimum over the division $m$. we complete the proof.
\end{proof}

\subsection{Upper bound of smallest eigenvalue under ND boundaries}
The left block with ND boundary conditions is just a shorter 1D RTIM, and many previous works have proven that its smallest eigenvalue follows a stretched exponential scaling \cite{PhysRevLett.69.534,PhysRevB.51.6411,PhysRevB.53.8486,PhysRevB.58.9131,FISHER1999222}. Here we present a proof for completeness, as some intermediate results will be used to numerically compute the upper bound $\Delta_b$ of the parity-restricted energy gap in the main text.

\begin{proposition}[Stretched-exponential bound for $\lambda^{ND}_1$]
At the critical point $h=h_c$, for any targeted failure probability $\epsilon>0$ there exists a constant $c>0$ independent of the system size $L$ such that $\lim_{L\to \infty} \mathrm{Pr}[\lambda^{ND}_1(L)\le e^{-c\sqrt{L}}]\ge 1-\epsilon$, where $\lambda^{ND}_1(L)$ is the smallest eigenvalue of the negative rate matrix $-T$ [Eq.~(\ref{eq:T_generator_ND})] with ND boundary conditions.
\label{thm:nd}
\end{proposition}

\begin{proof}
From the Courant–Fischer theorem [Eq.~(\ref{eq:CF_thm})] we have
\begin{equation}\label{eq:ND_CF}
\lambda_1^{ND}(L) \le \frac{\mathcal E(f,\,f)}{\langle f,\,f\rangle_\pi} \qquad (\forall f\neq 0).
\end{equation}
Therefore, we can bound $\lambda_1^{ND}(L)$ by constructing a suitable trial vector $f$.

Define $u_\ell\equiv\sum_{i=1}^{\ell-1}\ln(J_i/h)$. Then we have $\pi_k=\pi_L\,e^{2(u_L-u_k)}$.
For any $m\in\{1,\dots,L-1\}$, we define the suffix sums
\begin{equation}\label{eq:ND_suffix}
S_i \equiv \sum_{j=i}^{L}C_j^{-1},
\end{equation}
with $C_i=\pi_i p_i=\pi_{i+1}q_{i+1}$ and $C_L=\pi_L \kappa_L$ as defined before. Then we can construct the trial function $f$ with
\begin{equation}\label{eq:ND_trial}
f_i=\begin{cases}
1 & (1\le i<m)\\
S_i/S_m & (m\le i\le L).
\end{cases}
\end{equation}
Noting that $S_{i+1}-S_i=-C_i^{-1}$, we get
\begin{equation}\label{eq:ND_energy}
\mathcal E(f,\,f)=\frac{1}{S_m^2}\Big(\sum_{i=m}^{L-1}\frac{1}{C_i}+\frac{1}{C_L}\Big)=\frac{1}{S_m}.
\end{equation}

Moreover, observe that
\begin{equation}\label{eq:ND_norm}
\langle f,\,f\rangle_\pi \ge \sum_{i=1}^{m}\pi_i  = \pi_L\sum_{i=1}^{m} e^{2(u_L-u_i)},
\end{equation}
and
\begin{equation}
S_m=\frac{1}{h^2\pi_L} \sum_{j=m}^{L} e^{2(u_j-u_L)}.
\end{equation}
Hence we can bound
\begin{equation}\label{eq:ND_lambda_bound}
\lambda_1^{ND}(L) \le \frac{\mathcal E(f,\,f)}{\langle f,\,f\rangle_\pi}
 \le
\frac{h^2}{\sum_{i=1}^{m} e^{2(u_L-u_i)}\times \sum_{j=m}^{L} e^{2(u_j-u_L)}} = \frac{h^2}{\sum_{i=1}^{m} e^{-2u_i}\times \sum_{j=m}^{L} e^{2u_j}}.
\end{equation}
Note that $\sum_{i=1}^{m} e^{-2u_i}\ge e^{-2\min_{1\le i\le m}u_i}$ and $\sum_{j=m}^{L} e^{2u_j}\ge e^{2\max_{m\le j\le L}u_j}$, therefore we have
\begin{equation}\label{eq:ND_stretched_exp_exact}
\lambda_1^{ND}(L) \le  h^2 e^{-2(\max_{m\le j\le L}u_j - \min_{1\le i\le m}u_i)}.
\end{equation}
Let us define
\begin{equation}
U_L \equiv \frac{2\ln h}{\sqrt{L}} - 2 \frac{\max_{m\le j\le L} u_j -\min_{1\le i\le m}u_i}{\sqrt{L}},
\end{equation}
then we have $\ln \lambda_1^{ND}(L) \le U_L \sqrt{L}$.

To show the stretched exponential scaling of $\lambda_1^{ND}(L)$, we thus want to bound $U_L$. We can define $\xi_i\equiv \ln (J_i/h)$ such that $u_\ell = \sum_{i=1}^{\ell-1} \xi_i$. Note that we are considering an RTIM with $J_i$'s following i.i.d. distributions with a positive standard deviation $\mathrm{std}(\ln|J|)>0$. Therefore, at the critical point $\ln h_c = \overline{\ln |J|}$ we can regard $\xi_i$'s as independent random variables satisfying $\mathbb E[\xi_i]=0$ and $\mathrm{Var}(\xi_i)\equiv \delta^2>0$.

Let us further define right and left partial sums with a length of $r$ from the cutting point $m$
\begin{equation}\label{eq:ND_partial_sums}
S_r^+\equiv \sum_{i=m}^{m+r-1}\xi_i\quad(0\le r\le L-m),\qquad
S_r^-\equiv \sum_{i=m-r}^{m-1}\xi_i\quad(0\le r\le m-1).
\end{equation}
Note that the two partial sums involve different terms and are independent from each other.

Then we have
\begin{equation}\label{eq:ND_extrema}
\max_{m\le j\le L}u_j=u_m+\max_{0\le r\le L-m}S_r^+,\qquad
\min_{1\le i\le m}u_i=u_m-\max_{0\le r\le m-1}S_r^-,
\end{equation}
so that
\begin{equation}\label{eq:ND_UL_partial}
U_L=\frac{2\ln h}{\sqrt{L}}
-2\left(
\frac{\max_{0\le r\le L-m}S_r^+}{\sqrt{L}}
+\frac{\max_{0\le r\le m-1}S_r^-}{\sqrt{L}}
\right).
\end{equation}

The Donsker's theorem~\cite{durrett2019probability} shows that the partial sums in Eq.~(\ref{eq:ND_partial_sums}) converge to Brownian motions in the limit $L\to\infty$. Specifically, we can define two random processes
\begin{equation}\label{eq:ND_CLT}
X_t^{+} \equiv \frac{S_{\lfloor Lt\rfloor}^+}{\delta\sqrt{L}}\quad (0\le t \le 1-m/L),\qquad X_t^{-} \equiv \frac{S_{\lfloor Lt\rfloor}^-}{\delta\sqrt{L}} \qquad (0\le t \le (m-1)/L),
\end{equation}
which will converge to two independent standard Brownian motions in the limit $L\to\infty$ and $m/L \to \alpha$ ($0<\alpha<1$).

For a standard Brownian motion $W_t$, the Bachelier-L\'{e}vy formula~\cite{lerche1986boundary} shows that its first-passage time $T\equiv \inf\{t>0:W_t\ge c\}$ across a boundary at $c>0$ follows the distribution $\mathrm{Pr}(T\le t)=1-\Phi(c/\sqrt{t})+\Phi(-c/\sqrt{t})$, where $\Phi(x)$ is the cumulative distribution function of a standard normal distribution. Therefore we have
\begin{equation}
\mathrm{Pr}(\max_{0\le t\le \gamma}W_t \ge c) = \mathrm{Pr}(T\le\gamma) = 1 -\Phi(c/\sqrt{\gamma}) + \Phi(-c/\sqrt{\gamma}) =\mathrm{erfc}(c/\sqrt{2\gamma}),
\end{equation}
where $\mathrm{erfc}(x)\equiv \frac{2}{\sqrt{\pi}}\int_x^\infty e^{-t^2} dt$ is the complementary error function.
In other words, for any fixed $0<\alpha<1$ and any $c>0$, we have
\begin{equation}
\lim_{\substack{L\to \infty\\m=\lfloor \alpha L\rfloor}}\mathrm{Pr}\left(\frac{\max_{0\le r\le L-m}S_r^+}{\sqrt{L}}\ge c\right) = \lim_{L\to\infty}\mathrm{Pr}\left(\max_{0\le t\le 1-\alpha}X_t^+\ge \frac{c}{\delta}\right) = \mathrm{erfc}\left(\frac{c}{\delta\sqrt{2(1-\alpha)}}\right),
\end{equation}
and similarly
\begin{equation}
\lim_{\substack{L\to \infty\\m=\lfloor \alpha L\rfloor}}\mathrm{Pr}\left(\frac{\max_{0\le r\le m-1}S_r^-}{\sqrt{L}}\ge c\right) = \lim_{L\to\infty}\mathrm{Pr}\left(\max_{0\le t\le \alpha}X_t^-\ge \frac{c}{\delta}\right) = \mathrm{erfc}\left(\frac{c}{\delta\sqrt{2\alpha}}\right).
\end{equation}

Note that $\lim_{L\to\infty}(2\ln h/\sqrt{L})=0$. Combining the above results, we have for any fixed $0<\alpha<1$ and any $c>0$
\begin{equation}
\begin{aligned}
\lim_{\substack{L\to \infty\\m=\lfloor \alpha L\rfloor}} \mathrm{Pr} (U_L \le -c) &\ge \mathrm{Pr} \left(\frac{\max_{0\le r\le L-m}S_r^+}{\sqrt{L}} \ge c/4\right) \times \mathrm{Pr} \left(\frac{\max_{0\le r\le L-m}S_r^+}{\sqrt{L}} \ge c/4\right) \\
&= \mathrm{erfc}\left(\frac{c}{4\delta\sqrt{2(1-\alpha)}}\right) \mathrm{erfc}\left(\frac{c}{4\delta\sqrt{2\alpha}}\right) \equiv q(c).
\end{aligned}
\end{equation}
Observe that for any $\delta>0$ and $0<\alpha<1$, $q(c)$ is a continuous and monotonic function on $[0,\,+\infty)$ with $q(0)=1$ and $\lim_{c\to +\infty}q(c)=0$. Therefore, for any targeted failure probability $0<\epsilon<1$, we can choose a constant $c=q^{-1}(1-\epsilon)$ such that
\begin{equation}
\lim_{L\to \infty} \mathrm{Pr}[\lambda_1^{ND}(L) \le e^{-c\sqrt{L}}] \ge \lim_{\substack{L\to \infty\\m=\lfloor \alpha L\rfloor}} \mathrm{Pr} (U_L \le -c) \ge 1-\epsilon.
\end{equation}
\end{proof}

\subsection{Upper bound of smallest eigenvalue under DD boundaries}
Similar ideas can be used to bound the smallest eigenvalue for the right block with DD boundary conditions. Here the tridiagonal rate matrix in Eq.~(\ref{eq:T_generator_ND}) is modified into
\begin{equation}
T=\begin{pmatrix}
 -h^{2}-J_0^2 & h^{2} & 0 & \cdots & 0 \\
 J_1^{2} & -h^{2}-J_1^{2} & h^{2} & \ddots & \vdots \\
 \vdots & \ddots & \ddots & -h^{2}-J_{L-2}^{2} & h^{2} \\
 0 & \cdots & 0 & J_{L-1}^{2} & -h^{2}-J_{L-1}^{2}
\end{pmatrix},
\label{eq:T_generator_DD}
\end{equation}
with the right-jump rate $p_i\equiv T_{i,i+1}=h^2\mathbbm{1}_{\{i<L\}}$, the left-jump rate $q_i\equiv T_{i,i-1}=J_{i-1}^2\mathbbm{1}_{\{i>1\}}$, and the absorption rates $\kappa_1=J_0^2$ and $\kappa_L=h^2$ at the boundaries.

\begin{proposition}[Stretched-exponential bound for $\lambda_1^{DD}$]
At the critical point $h=h_c$, for any targeted failure probability $\epsilon>0$ there exists a constant $c>0$ independent of the system size $L$ such that $\lim_{L\to \infty} \mathrm{Pr}[\lambda^{DD}_1(L)\le e^{-c\sqrt{L}}]\ge 1-\epsilon$, where $\lambda^{DD}_1(L)$ is the smallest eigenvalue of the negative rate matrix $-T$ [Eq.~(\ref{eq:T_generator_DD})] with DD boundary conditions.
\label{thm:dd}
\end{proposition}
\begin{proof}
Again, from the Courant–Fischer theorem [Eq.~(\ref{eq:CF_thm})] we have
\begin{equation}\label{eq:DD_CF}
\lambda_1^{DD}(L) \le \frac{\mathcal{E}(f,\,f)}{\langle f,\,f\rangle_{\pi}}\qquad(\forall f\neq 0).
\end{equation}
Hence we can bound $\lambda_1^{DD}(L)$ by constructing a suitable trial vector $f$.

Again we define $\xi_i\equiv \ln (J_i/h)$ and $u_i = \sum_{k=1}^{i-1} \xi_i$ ($1\le i\le L$). We define $C_i=\pi_i h^2=\pi_{i+1}J_i^2$ ($1\le i \le L-1$) and $C_L=\pi_L h^2$ as defined before, but this time with the additional $C_0=\pi_1 J_0^2$. Then we define the prefix and the suffix sums
\begin{equation}\label{eq:DD_PS}
P_i \equiv \sum_{j=0}^{i-1}C_j^{-1},\qquad
S_i \equiv \sum_{j=i}^{L}C_j^{-1},
\end{equation}
so that $P_{i+1}-P_i=S_i-S_{i+1}=1/C_i$.

Now we split the chain into three parts by $1\le r<s\le L$ and consider a dent–shaped trial function $f$ with
\begin{equation}\label{eq:DD_trial}
f_i=\begin{cases}
P_i/P_r & (1\le i\le r)\\
1 & (r<i<s)\\
S_i/S_s & (s\le i\le L).
\end{cases}
\end{equation}
Direct calculation gives
\begin{equation}\label{eq:DD_energy_exact}
\mathcal{E}(f,\,f) = \frac{1}{P_r} + \frac{1}{S_s} \le \frac{h^2\pi_L}{\sum_{j=1}^{r-1}e^{2(u_j-u_L)}} + \frac{h^2\pi_L}{\sum_{j=s}^{L}e^{2(u_j-u_L)}},
\end{equation}
where in the inequality we have discarded the $1/C_0$ term in $P_r$.

On the other hand, discarding the function values $f_i$ at $i<r$ or $i>s$, we have
\begin{equation}\label{eq:DD_normLB}
\langle f,\, f\rangle_\pi \ge \sum_{i=r}^{s}\pi_i\ =\ \pi_L\sum_{i=r}^{s}e^{2(u_L-u_i)}.
\end{equation}
Combining Eqs.~(\ref{eq:DD_CF}–\ref{eq:DD_normLB}) yields
\begin{equation}\label{eq:DD_two_term}
\lambda_1^{\mathrm{DD}}(L) \le
\frac{h^2}{\sum_{i=r}^{s}e^{-2u_i}}\Bigg(
\frac{1}{\sum_{j=s}^{L}e^{2u_j}}+\frac{1}{\sum_{j=1}^{r-1}e^{2u_j}}
\Bigg).
\end{equation}
Again we can bound $\sum_{i=r}^{s}e^{-2u_i}\ge e^{-2\min_{r\le i\le s}u_i}$, $\sum_{j=s}^{L}e^{2u_j}\ge e^{2\max_{s\le j\le L}u_j}$ and $\sum_{j=1}^{r-1}e^{2u_j}\ge e^{2\max_{1\le j\le r-1}u_j}$, then we obtain
\begin{equation}\label{eq:DD_expbound}
\lambda_1^{\mathrm{DD}}(L) \le h^2\left[
e^{-2(\max_{s\le i\le L}u_i-\min_{r\le i\le s}u_i)}
 +
e^{-2(\max_{1\le i\le r-1}u_i-\min_{r\le i\le s}u_i)}
\right].
\end{equation}

We can further bound the exponents by independent terms as
\begin{equation}\label{eq:DD_AR}
\max_{s\le i\le L}u_i - \min_{r\le i\le s} u_i = \max_{0\le l\le L-s} \sum_{i=s}^{s+l-1}\xi_i + \max_{0\le l\le s-r} \sum_{i=s-l}^{s-1}\xi_i \ge \max_{0\le l\le L-s} \sum_{i=s}^{s+l-1}\xi_i \equiv A_R,
\end{equation}
and
\begin{equation}\label{eq:DD_AL}
\max_{1\le i\le r-1}u_i-\min_{r\le i\le s}u_i = \max_{1\le l\le r-1} \sum_{i=r-l}^{r-1}(-\xi_i) + \max_{0\le l\le s-r} \sum_{i=r}^{r+l-1}(-\xi_i) \ge \max_{1\le l\le r-1} \sum_{i=r-l}^{r-1}(-\xi_i) \equiv A_L,
\end{equation}
where for the inequalities we use the fact that the maximum is no less than the $l=0$ term.
Then we have $\lambda_1^{\mathrm{DD}}(L) \le 2h^2 e^{-2 \min\{A_R,A_L\}}$, and hence
\begin{equation}\label{eq:DD_logbound}
\frac{\ln\lambda_1^{DD}(L)}{\sqrt{L}}
 \le \frac{\ln 2 + 2\ln h}{\sqrt{L}}
 - \frac{2}{\sqrt{L}}\min\{A_R,A_L\} \equiv U_R.
\end{equation}

Similar as before, we compute the probability distributions in the limit $L\to \infty$ by mapping to Brownian motions. Specifically, for fixed $0<\alpha<\beta<1$, we set $r=\lfloor \alpha L\rfloor$ and $s=\lfloor \beta L\rfloor$. Then for any $c>0$, we have
\begin{equation}
\lim_{\substack{L\to \infty\\s=\lfloor \beta L\rfloor}}\mathrm{Pr}\left(\frac{A_R}{\sqrt{L}}\ge c\right) = \lim_{L\to\infty}\mathrm{Pr}\left(\max_{0\le t\le 1-\beta}W_t \ge \frac{c}{\delta}\right) = \mathrm{erfc}\left(\frac{c}{\delta\sqrt{2(1-\beta)}}\right),
\end{equation}
and
\begin{equation}
\lim_{\substack{L\to \infty\\r=\lfloor \alpha L\rfloor}}\mathrm{Pr}\left(\frac{A_L}{\sqrt{L}}\ge c\right) = \lim_{L\to\infty}\mathrm{Pr}\left(\max_{0\le t\le \alpha}W_t \ge \frac{c}{\delta}\right) = \mathrm{erfc}\left(\frac{c}{\delta\sqrt{2\alpha}}\right),
\end{equation}

Again, note that $\lim_{L\to\infty}(\ln 2 + 2\ln h)/\sqrt{L}=0$. Combining these results, we get for any fixed $0<\alpha<\beta<1$ and any $c>0$
\begin{equation}
\begin{aligned}
\lim_{\substack{L\to \infty\\r=\lfloor \alpha L\rfloor,\,s=\lfloor \beta L\rfloor}} \mathrm{Pr} (U_R \le -c) &= \mathrm{Pr} \left(\frac{A_R}{\sqrt{L}} \ge c/2\right) \times \mathrm{Pr} \left(\frac{A_L}{\sqrt{L}} \ge c/2\right) \\
&= \mathrm{erfc}\left(\frac{c}{2\delta\sqrt{2(1-\beta)}}\right) \mathrm{erfc}\left(\frac{c}{2\delta\sqrt{2\alpha}}\right) \equiv p(c).
\end{aligned}
\end{equation}
Again we note that for $\delta>0$ and $0<\alpha<\beta<1$, $p(c)$ is a continuous and monotonic function on $[0,\,+\infty)$ with $p(0)=1$ and $\lim_{c\to\infty}p(c)=0$.
Therefore, for any targeted failure probability $0<\epsilon<1$, we can choose a constant $c=p^{-1}(1-\epsilon)$ such that
\begin{equation}
\lim_{L\to \infty} \mathrm{Pr}[\lambda_1^{DD}(L) \le e^{-c\sqrt{L}}] \ge \lim_{\substack{L\to \infty\\r=\lfloor \alpha L\rfloor,\,s=\lfloor \beta L\rfloor}} \mathrm{Pr} (U_R \le -c) = 1-\epsilon.
\end{equation}
\end{proof}

\subsection{Completing the proof of the main theorem}
We are now equipped with all the relevant propositions, so the proof is straightforward.

\begin{proof}[Proof of Theorem~\ref{thm:main}]
Let $-T$ be the negative of the rate matrix in Eq.~(\ref{eq:T_generator_ND}) with eigenvalues $0\le \lambda_1(L)\le\lambda_2(L)\le\cdots$, and set $m\equiv\lfloor L/2\rfloor$. By Proposition~\ref{thm:nd} and Proposition~\ref{thm:dd}, at the critical point $h=h_c$, for any targeted failure probability $\epsilon>0$ there exist $c_L,\,c_R>0$ independent of the system sizes such that
\begin{equation}
\lim_{\substack{L\to \infty\\m=\lfloor L / 2\rfloor}} \Pr \left(\lambda^{ND}_1([1,\,m])\le e^{-c_L\sqrt{L}}\right)\ge \sqrt{1-\epsilon},
\end{equation}
and
\begin{equation}
\lim_{\substack{L\to \infty\\m=\lfloor L / 2\rfloor}} \Pr \left(\lambda^{DD}_1([m+1,\,L])\le e^{-c_R\sqrt{L}}\right)\ge \sqrt{1-\epsilon}.
\end{equation}
Then by taking $c_0=\min\{c_L,\,c_R\}$, Proposition~\ref{thm:block_upper_bound} yields
\begin{equation}
\lim_{L\to\infty}\Pr \left(\lambda_2(L)\le 2 e^{-c_0\sqrt{L}}\right) \ge 1-\epsilon.
\end{equation}
Note that the parity-restricted energy gap $\Delta_p(L)=\epsilon_1(L)+\epsilon_2(L)\le 2\epsilon_2(L)=4\sqrt{\lambda_2(L)}$. Therefore we have $\lim_{L\to\infty}\Pr \left(\Delta_p(L)\le 4\sqrt{2} e^{-(c_0/2)\sqrt{L}}\right) \ge 1-\epsilon$. Finally, since for any $\eta>0$ we have $4\sqrt{2}\le e^{\eta\sqrt{L}}$ for sufficiently large $L$, we can take any $0<c<c_0/2$ such that
\begin{equation}
\begin{aligned}
\lim_{L\to\infty}\Pr \left(\Delta_p(L)\le e^{-c\sqrt{L}}\right) &= \lim_{L\to\infty}\Pr \left(\Delta_p(L)\le e^{-(c_0/2)\sqrt{L}} e^{(c_0/2-c)\sqrt{L}}\right) \\
&\ge \lim_{L\to\infty}\Pr \left(\Delta_p(L)\le 4\sqrt{2}e^{-(c_0/2)\sqrt{L}}\right)\\
&\ge 1-\epsilon.
\end{aligned}
\end{equation}
\end{proof}

\section{Upper bound for parity-restricted energy gap}
The proof in Sec.~\ref{sec:proof} also gives us explicit formulas to upper bound the parity-restricted energy gap for any set of Ising couplings $J_i$'s and any transverse field $h$ at any system size $L$.
\begin{enumerate}
\item We can use Eq.~(\ref{eq:ND_stretched_exp_exact}) to bound $\lambda_1(L)$ and hence $\epsilon_1(L)=2\sqrt{\lambda_1(L)}$.
\item For any division at $m\in\{1,\,\cdots,\,L-1\}$, we can use Eq.~(\ref{eq:ND_stretched_exp_exact}) to bound $\lambda_1^{ND}([1,\,m])$ and use Eq.~(\ref{eq:DD_expbound}) to bound $\lambda_1^{DD}([m+1,\,L])$. Then we can use Eq.~(\ref{eq:two_blocks}) to bound $\lambda_2(L)$ and hence $\epsilon_2(L)=2\sqrt{\lambda_2(L)}$.
\item Finally we bound $\Delta_p(L)=\epsilon_1+\epsilon_2$ using the above bounds for $\epsilon_1$ and $\epsilon_2$.
\end{enumerate}

In principle any subdivisions $m$, $r$ and $s$ will work when using Eq.~(\ref{eq:ND_stretched_exp_exact}) and Eq.~(\ref{eq:DD_expbound}). However, to make the obtained bound tighter in numerical calculations, we further optimize over these subdivisions. The overall time cost is $O(L^3)$ for the largest system size $L$, and we note that many intermediate results can be reused for the numerical calculation of smaller system sizes to improve the efficiency.

%